\documentclass[a4paper,reqno,12pt]{amsart}
\usepackage[a4paper,margin=1in]{geometry}
\usepackage{graphicx}
\usepackage{amsmath,amssymb,amsthm,bm}
\usepackage{hyperref}

\theoremstyle{plain}
\newtheorem{theorem}{Theorem}[section]
\newtheorem{lemma}[theorem]{Lemma}
\newtheorem{proposition}[theorem]{Proposition}

\newtheorem{corollary}[theorem]{Corollary}
\theoremstyle{definition}
\newtheorem{remark}[theorem]{Remark}
\numberwithin{equation}{section}
\title[Global branches of Stokes waves of variable period on stratified fluids]
{Global branches of Stokes waves of variable period on stratified fluids}
\author{Vladimir Kozlov$^a$$^b$}

\address{$^a$Department of Mathematics, Link\"oping University, Link\"oping, Sweden}
\address{$^b$ Euler International Mathematical Institute (EIMI), Saint Petersburg, Russia}

\begin{document}
	
\begin{abstract}
We consider stratified steady water waves in a two dimensional channel. Our subject is branches of Stokes waves, bifurcating from laminar flows. We assume that the mass flux and the Bernoulli constant are fixed and consider the period of the wave as a parameter, which can change its value along the branch. A new class of density and Bernoulli functions is presented, for which laminar flows generate global bifurcation branches.
 The laminar flows are not necessary unidirectional and we show that the bifurcation branch can bifurcate from the laminar flow with arbitrary large period.

% bifurcating from the laminar flows.
%Our main
%subject is laminar solutions and the dispersion equation. Our main result is a %description of laminar solutions with dispersion equation having solutions.
% The main
%assumption is formulated in terms the Dirichlet problem for the Frechet %derivative
%calculated at the laminar solutions.

\end{abstract}

\maketitle

\section{Introduction}

We consider stratified steady water waves in a two-dimensional channel.
We use the classical formulation of the problem based on the Euler equations.
The surface tension is neglected and the water motion can be rotational.
Here we construct global branches of water waves. We assume that the flux and the Bernoulli constant are fixed and consider the period as a bifurcation parameter. We describe global branches of solutions bifurcating from a laminar flow.
We present a class of vorticity and density functions, where you can pick up a laminar flow whose dispersion equation has an arbitrary, a-priori fixed frequency as a solution. This class contains laminar flows with counter-currents.

\subsection{Statement of the main result}\label{SOkt11a}

Our object of study is two-dimensional stratified steady water waves  traveling
with constant speed $c$ under the influence of gravity. To eliminate the dependence on
time we use a moving reference frame, where the fluid occupies a  domain
 $$
D=D_\xi=\{(x,y)\,:\,-d<y<\xi(x),\;x\in\Bbb R\}
$$
in the channel with  the flat bottom $B$ given by $y=-d$ and with the free surface of
the flow $S=S_\xi$ given by $y=\xi(x)$. The density of the fluid $\rho$, defined
 in $\overline{D}$, is assumed to be positive and not necessarily constant.
 To describe a water wave we use the velocity of the flow $(u,v)$, the pressure $P$
 and the density $\rho$. If $\rho$ is non-constant the fluid is referred to as stratified and it is
stably stratified if $\rho$ is non-decreasing with depth.

 Corresponding relations describing the model can be found in the paper Walsh \cite{Wal} (see also the survey \cite{HSTW}, Sect.3). We recall these
 relations for readers convenience:
\begin{equation}\label{Se17a}
u_x+v_y=0\;\;\mbox{in $D$}\;\;\mbox{(incompressibility)},
\end{equation}
\begin{equation}\label{M3a}
(u-c)\rho_x+v\rho_y=0\;\;\mbox{in $D$}\;\;\mbox{(the conservation of mass)}
\end{equation}
and the conservation of momentum
\begin{eqnarray}\label{Se17aa}
&&(u-c)u_x+vu_y+\frac{P_x}{\rho}=0\;\;\mbox{in $D$},\nonumber\\
&&(u-c)v_x+vv_y+\frac{P_y}{\rho}=-g\;\;\mbox{in $D$},
\end{eqnarray}
where
 $g$ is the gravitational constant.
The boundary conditions are
\begin{equation}\label{Sep18a}
v=(u-c)\xi_x\;\;\mbox{and}\;\;P=P_{atm}\;\;\mbox{on $S$},
\end{equation}
where $P_{atm}$ is the atmospheric pressure, and
\begin{equation}\label{Sep18aa}
v=0\;\;\mbox{on $B$}.
\end{equation}

The pseudostream function $\psi=\psi(x,y)$ is defined by
$$
\psi_x(x,y)=-\sqrt{\rho}v(x,y),\,\;\psi_y(x,y)=\sqrt{\rho}(u(x,y)-c).
$$
Then  equations (\ref{Se17a}) and  (\ref{M3a}) are satisfied if $\rho$ is a function of $\psi$.
In what follows we assume that $\rho=\rho(-\psi)$ and $\rho(p)$ is a positive function.
%\begin{equation}\label{Ma13aa}
%\rho_p(p)\leq 0.
%\end{equation}

The  relative pseudomass flux is defined by
$$
p_0=\int_{-d}^{\xi(x)}\sqrt{\rho}(u(x,y)-c)dY
$$
and it does not depend on $x$. We will assume that $p_0<0$.

The energy in the system
\begin{equation}\label{J25a}
E:=\frac{1}{2}(\psi_x^2+\psi_y^2)+ P+g\rho y
\end{equation}
is constant along the stream lines of $\psi$. This allows  to define
 the Bernoulli function
\begin{equation}\label{Ma13b}
\beta(\psi)=\frac{d E}{d\psi}.
\end{equation}
It can be verified that
\begin{equation}\label{Ma13ba}
\beta (\psi)=\Delta\psi-gy\rho_p(-\psi)=0\;\;\mbox{in $D$}.
\end{equation}
Boundary conditions for $\psi$ are
\begin{equation}\label{Ma13bb}
\psi(x,\xi(x))=0\;\;\mbox{and}\;\;\psi(x,-d)=-p_0
\end{equation}
together with the Bernoulli boundary condition
\begin{equation}\label{Ma13bc}
\frac{1}{2}|\nabla\psi|^2+g\rho(0)(\xi(x)+d)=R\;\;\mbox{for $y=\xi(x)$},
\end{equation}
which is obtained from (\ref{J25a}) by setting $y=\xi(x)$ and using that
$\psi(x,\xi(x))=0$. The constant $R$ is called the Bernoulli constant.

We set
\begin{equation}\label{Okt5a}
\omega(y,\psi)=-gy\rho'(-\psi)-\beta(\psi),
\end{equation}
where $\rho'(-\psi)=-\partial_\psi\rho(-\psi)$.
Then equation (\ref{Ma13ba}) takes the form
\begin{equation}\label{Okt5aa}
\Delta\psi+\omega(y,\psi)=0\;\;\mbox{in $D$}.
\end{equation}
In the formulation (\ref{Ma13ba})--(\ref{Ma13bc}), we assume that the constants $p_0$ and $R$ are given and the functions $\psi$ and $\xi$ are unknowns.

Our main assumption in our study concerns laminar solutions, i.e.
solutions of (\ref{Okt5aa}) depending only on $y$. A laminar solution $\Psi(y)$ satisfies
\begin{equation}\label{J18aa}
\Psi^{''}(y)+\omega(y,\Psi(y))=0\;\;\mbox{on $(-d,0)$}
\end{equation}
together with the boundary conditions
\begin{equation}\label{J18ab}
\Psi(0)=0\;\;\mbox{and}\;\;\Psi(-d)=-p_0,
\end{equation}
\begin{equation}\label{J18ac}
\frac{1}{2}|\Psi'(0)|^2+g\rho(0)d=R.
\end{equation}
Here $\Psi$ and $d>0$ are considered as unknowns.
We assume that $\rho$ and $\beta$  are $C^{2,\alpha}$ and $C^{1,\alpha}$ functions respectively, for a certain $\alpha\in (0,1)$, satisfying
\begin{equation}\label{F11a}
|\rho(p)|+|\beta(p)|\leq C_1(1+|p|)\;\;\mbox{and}\;\;|\rho(p_2)-\rho(p_1)|+|\beta(p_2)-\beta(p_1)|\leq C_1|p_2-p_1|.
\end{equation}

Our main assumption is the following

\noindent
{\bf Assumption (A)} The spectral problem
\begin{eqnarray}\label{J18c}
&&w^{''}(y)+\omega_p(y,\Psi(y))w+\mu w=0\;\;\mbox{on $(-d,0)$},\nonumber\\
&&w(-d)=w(0)=0
\end{eqnarray}
has only positive eigenvalues $\mu$ for each $d>0$, where $\Psi$ is a solution of the equation (\ref{J18aa}) with $\Psi(0)=0$. Here
\begin{equation}\label{Ma9a}
\omega_p((y,\Psi)=-\beta'(\Psi)+y\rho^{''}(-\Psi)
\end{equation}
i.e. it is the derivative of $\omega(y,p)$ with respect to $p$.

We consider the following branch of the laminar solutions $(\Psi(y;s),d(s))$, which solve the problem
\begin{eqnarray}\label{F19a}
&&\Psi^{''}(y)+\omega(y,\Psi(y))=0\;\;\mbox{on $(-\infty,0]$},\nonumber\\
&&\Psi(0)=0,\;\;\Psi'(0)=s,
\end{eqnarray}
together with
\begin{equation}\label{A10a}
\Psi(d(s))=-p_0.
\end{equation}
Due to (\ref{F11a}) the problem (\ref{F19a}) is uniquely solvable and we denote its solution by $\Psi(y;s)$. We take the maximal interval $(-\infty,s_*)$ such that equation (\ref{A10a}) holds and $\Psi_y(d(s);s)<0$ for $s<s_*$. Clearly the pair $(\Psi(y;s),d(s))$ solves the problem (\ref{J18aa})-(\ref{J18ac}) with
\begin{equation}\label{A17a}
R={\mathcal R}(s):=\frac{1}{2}\Psi_y(0;s)^2+g\rho(0)d(s).
\end{equation}

The dispersion equation is expressed through the solvability of the spectral problem (see Sect. \ref{SA15a})
\begin{eqnarray}\label{F21a}
&&(\partial_y^2-\tau^2)u+\omega_p(y,\Psi)u=0
\;\;\mbox{in $(-d,0)$},\nonumber\\
&&\Psi_yu_y-\kappa u=0\;\;\mbox{for $y=0$},\nonumber\\
&& u(-d)=0,
\end{eqnarray}
where
\begin{equation}\label{F21aa}
\kappa=\frac{\Psi_y(0)\Psi_{yy}(0)+g\rho(0)}{\Psi_y(0)}.
\end{equation}
We are looking for a nontrivial solution $u$ and a positive number $\tau$. So we may assume that $u(0)=1$. Having this in mind, we introduce the function $\gamma(y;\tau)$, which is the solution of
\begin{eqnarray}\label{F21ab}
&&(\partial_y^2-\tau^2)\gamma+\omega_p(y,\Psi)\gamma=0\;\;\mbox{in $(-d,0)$},\nonumber\\
&&\gamma(0)=1,\nonumber\\
&& \gamma(-d)=0.
\end{eqnarray}
From the boundary condition for $y=0$ in (\ref{F21a}), we obtain the following scalar  equation for $\tau$:
\begin{equation}\label{F21ac}
\sigma(\tau;s):=(\Psi_y\gamma_y-\kappa )|_{y=0}=0,
\end{equation}
which is called the dispersion equation.
The corresponding eigenfunction is $\gamma(y;\tau)$. The function $\sigma$ is strongly increasing with respect to $\tau$ and tends to infinity as $\tau\to \infty$. Therefore existence of a positive root to (\ref{F21ac}) is equivalent to
$\sigma(0;s)<0$. In Sect. \ref{SA14a} we prove that
\begin{equation}\label{F13ab}
\sigma(0;s)=\frac{-s-g\rho_0\dot{d}(s)-s\dot{d}V_y(0)}{s\dot{d}(s)}=
\frac{-{\mathcal R}'(s)-s\dot{d}V_y(0)}{s\dot{d}(s)},
\end{equation}
where $V$ is the solution of the problem
\begin{eqnarray}\label{F13a}
&&V^{''}(y)+\omega_p(y,\Psi(y;s))V=\rho_p(-\Psi)\;\;\mbox{on $(-d,0)$}\nonumber\\
&&V(-d)=V(0)=0
\end{eqnarray}
and hence
\begin{equation}\label{A18a}
V_y(0)=\int_{-d}^0\gamma(z;0)\rho'(-\Psi(z;s))d\,z.
\end{equation}
\begin{proposition}
The dispersion equation (\ref{F21ac}) has a positive solution for a certain $s<0$ in the cases

(i) The density  $\rho$ is stably stratified, i.e. $\rho'(p)\leq 0$;

(ii) $s_*\geq 0$ or $s_*$ is a small negative number.

(iii) $g\rho_0+s_*V_y(0;s*)>0$.
\end{proposition}

In order to formulate the main theorem, let us introduce some functional spaces.
For $k=0,1,\ldots$, $\alpha\in(0,1)$ and $-\infty<a<b<\infty$,
we denote by $C^{k,\alpha}(\overline{D}_{a,b})$ and $C^{k,\alpha}([a,b])$
the H\"older spaces of functions in $\overline{D}_{a,b}=\{(x,y)\in \overline{D}
\,:\,a\leq x\leq b\}$ and $[a,b]$ respectively.
% By $C^{k,\alpha}_0(\overline{\Omega}_{a,b})$ we denote the subspace of
%functions in $C^{k,\alpha}(\overline{\Omega}_{a,b})$ vanishing for $y=-d$.

The space $C^{k,\alpha}_{\textrm{b}}(\overline{D})$ consists of functions
$u$ defined on $\overline{D}$
such that
$$
\sup_{a\in \Bbb R}||w||_{C^{k,\alpha}(\overline{D}_{a,a+1})}<\infty.
$$
 The subspace   $C^{k,\alpha}_{\Lambda}(\overline{D})$ in $C^{k,\alpha}_{\textrm{b}}(\overline{D})$
consists of $\Lambda$-periodic functions and the subspace $C^{k,\alpha}_{0,e,\Lambda}(\overline{D})$ consists of   $\Lambda$--periodic, even functions vanishing for $y=\xi(x)$.

Similarly, one can define the spaces $C^{k,\alpha}_{\textrm{b}}(\Bbb R)$,
$C^{k,\alpha}_{\Lambda}(\Bbb R)$ and
$C^{k,\alpha}_{e,\Lambda}(\Bbb R)$.

In the just introduced functional spaces we shall use the norms
$$
||u||_{C^{k,\alpha}_{\Lambda}(\overline{D})}=
||u||_{C^{k,\alpha}_{\Lambda}(\overline{D}_{-\Lambda/2,\Lambda/2})}.
$$
Similarly, we define the norms in other spaces of $\Lambda$-periodic functions.

We always assume that $\psi_y\neq 0$ on the free surface $S_\xi$.

\begin{theorem}\label{TA17b} Let the laminar flow $(\Psi,d)$ is chosen such that, $s<0$ and the dispersion  equation (\ref{F21ac}) has a solution $\tau_*>0$ with the corresponding eigenfunction $\gamma(y;\tau_*)$, and let $R$ be given by (\ref{A17a}).
Then there exist $\varepsilon > 0$ and a mapping\footnote{Here in what follows we will use the notation $F[t]$, $F[t](x,y)$, where $t$ is a parameter on the branch of solutions and $(x,y)$ are variables in the domain}
\begin{equation}\label{A17aa}
(-\varepsilon,\varepsilon)\ni t\rightarrow (\psi[t],\xi[t],\Lambda(t)))\in C^{2,\alpha}_{0,e,\Lambda(t)}(\overline{D_\xi})\times C^{2,\alpha}_{e,\Lambda(t)}(\Bbb R)\times\Bbb R_+,
\end{equation}
where $(\Psi[t],\xi[t])$ is a $\Lambda(t)$--periodic, even solution to (\ref{J18aa})--(\ref{J18ac}) such that the vector function
\begin{equation}\label{A17ab}
(-\varepsilon,\varepsilon)\ni t\rightarrow (\lambda(t),\xi[t](\lambda^{-1}\hat{x})))\in \Bbb R_+\times C^{2,\alpha}_{e,\Lambda_*}(\Bbb R),
\end{equation}
where
\begin{equation}\label{A17b}
\Lambda_*=\frac{2\pi}{\tau_*},\:\:\lambda(t)=\frac{\Lambda_*}{\Lambda(t)},\;\;\hat{x}=\lambda x,
\end{equation}
is continuous together with its first derivative with respect to $t$. Furthermore,
\begin{equation}\label{A17ac}
\Lambda(0)=\Lambda_*,\;\;\xi[t](x)=-\frac{t}{\Psi_{*y}(0)}
\cos(\tau_*\lambda(t)x)+o(|t|)
\end{equation}
and
\begin{equation}\label{A17acz}
\psi[t](x,y)=\Psi_*(y)+t\gamma(y;\tau_*)\cos(\tau_*\lambda(t)x)
+o(|t|).
\end{equation}

\end{theorem}
 In what follows we shall use the notations $\tau_*$ and $\Lambda_*$ introduced in the theorem and the corresponding laminar solution we denote by $(\Psi_*(y),d_*)$. The parameter $s$, which is included in the definition of the laminar solution, is fixed and we will not indicate the dependence of the function on $s$, if it is not needed.

\begin{remark}\label{RA25a}
1). If the functions $\rho$ and $\beta$ are real analytic then the functions $\Lambda(t)$ and $\xi[t]$ are also real analytic. This can be prove by complementing
the Crandall--Rabinowitz theorem by results from Chapter I.16, \cite{Ki1}.

2). Smoothness properties of the function $\psi[t]$ will be clarified latter. It can be obtain also from the analysis of the boundary value problem (\ref{J18aa})--(\ref{J18ac}) if we know smoothness properties of the functions $\Lambda(t)$ and $\xi[t]$.
\end{remark}

In order to formulate a global version of the above theorem,  we introduce the set
\begin{equation}\label{A30ac}
{\mathcal U}_\delta =\{(\psi,\xi,\Lambda)\;:\;\xi\in C^{2,\alpha}_{e,\Lambda}(\Bbb R),\;
\psi\in C^{2,\alpha}_{0,e,\Lambda}(D_\xi),\;\;\Lambda\in (0,\infty)\}
\end{equation}
for $\delta >0$,
such that $\psi$ and $\xi$ are $\Lambda$--periodic and even functions satisfying
\begin{equation}\label{A30aaa}
\delta\leq\Lambda\leq\delta^{-1},
\end{equation}
\begin{equation}\label{A30ab}
\max_x\xi(x)\leq R-\delta\;\;\mbox{on $S_\xi$},
\end{equation}
\begin{equation}\label{A30a}
\max_{(x,y)} |\psi(x,y)|\leq \delta^{-1},
\end{equation}
\begin{equation}\label{A30aa}
\max_x|\xi'(x)|\leq \delta^{-1},
\end{equation}
\begin{equation}\label{Ma1a}
\psi_y\leq -\delta\;\;\mbox{on $S_\xi$}
\end{equation}
and
\begin{equation}\label{A30ad}
\min_x (\xi(x)+d)\geq \delta.
\end{equation}
Let $\varepsilon$ be another positive number. We will consider $\psi$ satisfying $\psi(x,\xi(x))=0$ and $\psi_y(x,\xi(x))<0$. This motivates the introduction of the following set.

\noindent
{\bf Definition of ${\mathcal U}_{\delta,\varepsilon}$}\label{Ma30a} Assume that $(\psi,\xi,\Lambda)\in {\mathcal U}_\delta$. Then for each $x\in\Bbb R$ there exists $y_*\in [-d,\xi(x))$, such that either $-\psi_y(x,y)>\delta$ on $(y_*(x),\xi(x))$ and this interval is the largest, i.e. $-\psi_y(x,y_*(x))=\delta$ (in this case $y_*(x)<-d$) or $y_*(x)=-d$. We define
\begin{equation}\label{Jun10a}
\varepsilon_*(\psi)=\min_x \psi(x,y_*(x)).
\end{equation}
Let also $x_*$ be a point, where the minimum is attained.
Observe that this is a positive function for all $\psi\in{\mathcal U}_\delta$.
 The function $\psi$ belongs to ${\mathcal U}_{\delta,\varepsilon}$ if and only if $\varepsilon_*(\psi)>\varepsilon$. If $y_*(x)=-d$ for all $x$, then $\psi(x,y_*(x))=-p_0$ in the case when $(\psi,\xi,\Lambda)$ satisfies the problem (\ref{Ma13ba})--(\ref{Ma13bc}).

\medskip
The important property, which follows from the definition, is the following
\begin{equation}\label{Jun14a}
\int_{y_*(x)}^{\xi(x)}\psi_y(x,s)ds=-\psi(x,y_*(x)),
\end{equation}
and therefore
\begin{equation}\label{Jun14aa}
-\psi(x_*,y_*(x_*))>\delta (\xi(x_*)-y_*(x_*)).
\end{equation}
One can check that
$$
{\mathcal U}_{\delta_1,\varepsilon_1}\subset {\mathcal U}_{\delta_2,\varepsilon_2}\;\;\mbox{if}\;\;\delta_2\leq \delta_1\;\;\mbox{and}\;\;\varepsilon_2\leq\varepsilon_1
$$
and that
\begin{equation}\label{Jun10b}
{\mathcal U}_\delta=\bigcup_\varepsilon {\mathcal U}_{\delta,\varepsilon}.
\end{equation}
Let also
$$
\widehat{\mathcal U}_{\varepsilon}=\bigcup_{\delta>0}{\mathcal U}_{\delta,\varepsilon},
$$
\begin{equation}\label{A30b}
{\mathcal U}=\bigcup_\delta {\mathcal U}_\delta
\end{equation}
and
\begin{equation}\label{A30ba}
{\mathcal O}=\{(\psi,\xi,\Lambda)\in {\mathcal U}\;:\; (\psi,\xi,\Lambda)\;\;
\mbox{satisfy the problem}\;\; (\ref{J18aa})--(\ref{J18ac})\}.
\end{equation}

\begin{theorem}\label{TMa2a} We assume that all assumptions in Theorem \ref{TA17b} are fulfilled and the functions $\rho$ and $\beta$ are real analytic. Then there exists the branch
\begin{equation}\label{Ma2a}
\Bbb R\ni t\rightarrow (\Psi[t],\xi[t],\Lambda(t)))\in C^{2,\alpha}_{0,e,\Lambda(t)}(\overline{D_\xi})\times C^{2,\alpha}_{e,\Lambda(t)}(\Bbb R)\times\Bbb R_+,
\end{equation}
where $(\Psi[t],\xi[t])$ is a $\Lambda(t)$--periodic solution to (\ref{J18aa})--(\ref{J18ac}) coinciding with the branch (\ref{A17aa}) for small $t$. The vector function
\begin{equation}\label{Ma2aa}
\Bbb R\ni t\rightarrow (\lambda(t),\xi[t](\lambda^{-1}(t)\hat{x}))\in \Bbb R_+\times C^{2,\alpha}_{0,e,\Lambda^*}(\Bbb R)
\end{equation}
is real analytic with respect to $t$ (up to a local reparametrization). Furthermore, one of the following alternatives occurs:

(a) there exists $\varepsilon_*>0$  such all element of the branch belongs to $\widehat{\mathcal U}_{\varepsilon_*}$ and that for every $\delta>0$ there exists
$t_\delta$ such that all elements  $(\psi[t],\xi[t],\Lambda(t))$  do not belong to ${\mathcal U}_{\delta,\varepsilon_*}$ for $|t|>t_\delta$;

b) There exist  sequences $\{t_j\}$ approaching $\infty$  and $\{\delta_j\}$ approching zero as $j\to\infty$, such that
$(\psi[t_j],\xi[t_j],\Lambda(j))\in {\mathcal U}_{\delta_j}$ and
\begin{equation}\label{Jun7b}
 \psi[t_j](x_{*j},y_{*j})\to 0, \;\; \psi_y[t_j](x_{*j},y_{*j})=\delta_j \;\;\mbox{as}\;\;j\to\infty.
\end{equation}
Furthermore,
\begin{equation}\label{Jun14b}
\psi[t_j](x_{*j},y_{*j})>\delta_j(\xi[t_j](x_{*j})-y_{*j}(x_{*j}))\;\;\mbox{for $j>0$}.
\end{equation}

\end{theorem}

We note that from (\ref{Jun14b}) it follows that the product $\delta_j$ and $\xi[t_j](x_{*j})-y_{*j}(x_{*j})$ tends to zero when $j\to\infty$.

\section{Laminar flows and the dispersion equation}

\subsection{The laminar flows}

In this section we will study laminar solutions, i.e. solutions to the problem (\ref{J18aa})--(\ref{J18ac}). %and the corresponding dispersion equations.
The pair $(\Psi,d)$ is considered as an unknown.
We assume that $p_0<0$ and the functions $\rho$ and $\beta$ satisfy (\ref{F11a}).
Under this condition the Cauchy problem (\ref{F19a})
has a unique solution for each real $s$, which is denoted by $\Psi(y;s)$.

\begin{lemma}\label{LemF11a} There exists $d_1>0$ such that the problem (\ref{J18aa}), (\ref{J18ab})
has a unique solution for each $d\leq d_1$. This solution satisfies
\begin{equation}\label{F11aa}
\Psi(y)=\frac{p_0y}{d}+w,\;\;|w|d^{-1}+|w'|\leq Cd\;\;\mbox{for $y\in [0,-d]$},
\end{equation}
where $C$ is independent of $y$.

\end{lemma}
\begin{proof} First, consider  the problem
\begin{equation}\label{J18b}
u^{''}(y)=f(y)\;\;\mbox{om $(-d,0)$}, \;\;u(-d)=u(0)=0.
\end{equation}
Its solution is given by
\begin{equation}\label{J18ba}
u(y)=\int_{-d}^0G(y,z)f(z)dz,
\end{equation}
where $G$ is the Green function:
$$
G(y,\tau)=\frac{1}{d}y(\tau+d)\;\;\mbox{for $\tau<y$ and}\;\;G(y,\tau)=\frac{1}{d}\tau(y+d)\;\;\mbox{for $\tau>y$}.
$$
This function is negative inside $(-d,0)^2$ and $|G|\leq d$.

We are looking for the solution to  (\ref{J18aa}), (\ref{J18ab}) in the form
$\Psi(y)=\frac{p_0y}{d}+w(y)$.
Then $w$ must satisfy
$$
w^{''}(y)+\omega(y,\frac{p_0y}{d}+w(y))=0\;\;\mbox{on $(-d,0)$}
$$
and $w(-d)=w(0)=0$. Applying (\ref{J18ba}), we get
$$
w(y)=F(y,w):=\int_{-d}^0G(y,\tau)f(\tau,w(\tau))d\tau,
$$
where
$$
f(y,w(y))=-\omega(y,\frac{p_0y}{d}+w(y)).
$$
Using (\ref{F11a}), we can verify that
$$
|F(y,w_2)-F(y,w_1)|\leq Cd^2||w_2-w_1||_{L_\infty}´.
$$
Therefore $F$ is a contraction operator for small $d$ and $|F(y,w)|\leq cd^2$.
Applying Banach fixed point theorem, we obtain existence of a unique solution $w$ for small $d$ satisfying (\ref{F11aa}).

\end{proof}

\begin{remark} Let us construct solutions to the problem (\ref{F19a}) for large $|s|$, $s<0$. We are looking for the solution in the form
$$
\Psi(y;s)=sy+w,
$$
where $w$ solves
\begin{eqnarray}\label{M3az}
&&w^{''}(y)+\omega(y,sy+w)=0\;\;\mbox{on $(-d,0)$},\nonumber\\
&&\Psi(0)=0,\;\;\Psi'(0)=0.
\end{eqnarray}
Then
\begin{equation}\label{M3aa}
w(y)=F(w)(y):=-\int_y^0(z-y)\omega(z,sz+w(z))dz.
\end{equation}
Since
$$
|F(w_1)(y)-F(w_2)(y)|\leq Cd^2\max_{-d\leq z\leq 0} |w_1(z)-w_2(z)|,
$$
the operator $F$ is a contraction in $L^\infty(-d,0)$ for small $d$ and the solution to (\ref{M3aa}) satisfies
$$
|w(y;s)|\leq C d^2(1+|s|d),\;\;|w'(y;s)|\leq C d(1+|s|d)\;\;\mbox{and}\;\;
|\partial_s w(y;s)|\leq C d^3(1+|s|d)^2.
$$
Therefore, if $|s|\geq C_1d^{-1}$, $s<0$, then
$$
|w(y;s)|\leq Cd^2,\;\;|w'(y;s)|\leq Cd\;\;\mbox{and}\;\;|\partial_s w(y;s)|\leq Cd^3
$$
for $y\in [-d,0]$.
\end{remark}

In the previous proposition {\bf Assumption (A)} is not used but it
 is essential in proving of the next proposition.

\begin{proposition}\label{PM1a} Solutions to the problem (\ref{F19a}) satisfy the following monotonicity property
\begin{equation}\label{F7b}
\Psi(y,s_2)>\Psi(y,s_1)\;\;\mbox{for $y<0$ if $s_2<s_1$}.
\end{equation}
\end{proposition}
\begin{proof} Since $s_2<s_1$ the inequality (\ref{F7b}) is true for small $y$,
Assume that $\Psi(y,s_2)=\Psi(y,s_1)$ for a certain $y_0$, which is the largest negative $y$ with this property. Let $s_3=(s_2+s_1)/2$. Then there exists $y_1>y_0$ such that $\Psi(y_1,s_3)=\Psi(y_1,s_1)$ or $\Psi(y_1,s_3)=\Psi(y_1,s_2)$ and this is the largest such $y$. We put $\tilde{s}_1=s_2$, $\hat{s}_1=s_3$ or $\tilde{s}_1=s_3$, $\hat{s}_1=s_1$ depending on the above choice of $s_3$. Then we have $\Psi(y_1,\tilde{s}_1)=\Psi(y_1,\hat{s}_1)$ and  $\tilde{s}_1<\hat{s}_1$.

Continuing this procedure we can construct a sequences $\{\hat{s}_j\}$, $\{\tilde{s}_j\}$ and $\{y_j\}$ with the following properties
\begin{equation}\label{M7a}
s_2\leq\hat{s}_1\leq\hat{s}_2\cdots <\cdots \tilde{s}_2\leq\tilde{s}_1\leq s_1,\;\; y_1<y_2<\cdots <0
\end{equation}
and
 $$
 \Psi(y_j,\hat{s}_j)=\Psi(y_j,\tilde{s}_j)\;\;
 \mbox{and}\:\;\Psi(y,\hat{s}_j)<\Psi(y,\tilde{s}_j)\;\;\mbox{ for $y\in (y_j,0)$}.
 $$
Moreover
\begin{equation}\label{M7aa}
\tilde{s}_j-\hat{s}_j\to 0\;\;\mbox{ as $j\to\infty$}.
\end{equation}
We denote by $s_\dag$ the limit point of the sequences $\{\hat{s}_j\}$ and $\{\tilde{s}_1\}$, which is the same due to (\ref{M7a}) and (\ref{M7aa}). Let also $y_\dag$ the limit of the sequence $\{y_j\}$.

If $y_\dag=0$ then due to uniqueness in Lemma \ref{LemF11a} there are no two different solution coinciding for a small $|y|$. If $y_\dag>0$ then application of the assumption
(A) to the function $\Psi(\cdot,s_\dag)$ on the interval $(y_\dag,0)$ implies that the Dirichlet problem  for the Fr\'{e}chet derivative at the function $\Psi(\cdot,s_\dag)$ has positive first eigenvalue.

Furthermore, we have
\begin{eqnarray*}
&&\omega(y,\Psi(y;\hat{s}_j))-\omega(y,\Psi(y;\tilde{s}_j))=
\int_0^1\frac{d}{d\,t}\omega(y,t\Psi(y;\hat{s}_j)+(1-t)\Psi(y;\tilde{s}_j))dt\\
&&=\int_0^1\omega_p(y,t\Psi(y;\hat{s}_j)+(1-t)\Psi(y;\tilde{s}_j))dt
(\Psi(y;\hat{s}_j)-\Psi(y;\tilde{s}_j))
\end{eqnarray*}
where
$$
A_j(y)=\int_0^1\omega_p(y,t\Psi(y;\hat{s}_j)+(1-t)\Psi(y;\tilde{s}_j))dt.
$$
Since both $\Psi(y;\hat{s}_j)$ and $\Psi(y;\tilde{s}_j)$ approach $\Psi(y;s_\dag)$ when $j\to\infty$ on the interval $[M,0]$ for a fixed $M$, $M<y_j$, we have that
$$
|A_j(y)-\omega_p(y,\Psi(y;s_\dag))|\leq \varepsilon_j\;\;\mbox{on $[M,0]$},
$$
where $\varepsilon_j\to o$ as $j\to\infty$. This implies that $\Psi(y;\hat{s}_j)$ must coincide with $\Psi(y;\tilde{s}_j)$ for large $j$. This contradiction proves Lemma.

\end{proof}

Consider the equation
\begin{equation}\label{M1a}
\Psi(y;s)=-p_0.
\end{equation}
The largest negative $y$ satisfying the equation we denote by $-d(s)$. According to Lemma \ref{LemF11a} and Proposition \ref{PM1a} the equation is solvable for $s$ in a neighborhood of $-\infty$, and  $d(s)\to 0$ and $\Psi'(-d(s);s)=p_0/d(s)+O(1)$ as $s\to-\infty$. The function $s\to d(s)$ can be continued up to the first $s_*$ such that $\Psi'(-d(s_*);s_*)=0$ if $s_*$ is finite or $s_*=+\infty$. Therefore, $\Psi'(-d(s_*);s_*)<0$ for $s\in (-\infty, s_*)$.

In the next proposition and in what follows we shall denote the derivative with respect to $s$ by $\dot{F}(y;s)=\partial_s F(y;s)$.

\begin{proposition}\label{PF20a} (i) $\dot{\Psi}(y;s)<0$ for all $s\in\Bbb R$ and for all $y<0$;

(ii) If $s_*$ and $d(s_*)$ are finite then
 $\dot{d}(s)\to\infty$ as $s\to s_*$.

\end{proposition}
\begin{proof} (i) We have that $\dot{\Psi}(0;s)=0$ and $\Psi_{ys}(0;s)=1$ for all $s<s_*$. Assume that
$\Psi_s(y_1;s)=0$ for a certain $y_1<0$. Then the function $w(y)=\Psi_s(y;s)$ satisfies the problem (\ref{J18c}) on the interval $(y_1,0)$ and by this assumption $\Psi_s=0$ on the interval $(y_1,0)$ and hence everywhere. This contradicts to
$\Psi_{ys}(0;s)=1$.

(ii) Since $\Psi(-d(s_*),s_*)=0$, we have
$$
-\Psi'(-d(s_*),s_*)\dot{d}(s_*)+\dot{\Psi}(-d(s_*),s_*)=0.
$$
This implies the required assertion.
\end{proof}

%We introduce the function
%\begin{equation}\label{G8aa}
%{\mathcal R}(s)=\frac{1}{2}s^2+g\rho(0)d(s).
%\end{equation}
%Let ${\mathcal R}(s)$ be given by (\ref{A17a}). Then
%$$
%{\mathcal R}(s)\to\infty\;\;\mbox{as $s\to 0$ or $s\to -\infty$}.
%$$
If $(\Psi(y;s),d(s))$ solves the problem
\begin{eqnarray}\label{F19az}
&&\Psi^{''}(y)+\omega(y,\Psi(y))=0\;\;\mbox{on $(-d(s),0]$},\nonumber\\
&&\Psi(0)=0,\;\;\Psi'(0)=s,
\end{eqnarray}
then it solves also the problem (\ref{J18aa})--(\ref{J18ac}) with
$
R={\mathcal R}(s).
$

\subsection{Green's function. Generalized maximum principle.}\label{SM2a}

Consider the Dirichlet boundary value problem
\begin{eqnarray}\label{F23a}
&&U^{''}(y)+H(y)U=f(y)\;\;\mbox{on $(y_1,0)$},\nonumber\\
&&V(y_1)=\alpha_1,\;\;V(0)=\alpha_2,
\end{eqnarray}
where $H$ is a bounded function. We assume that the least eigenvalue of the Dirichlet problem for the operator $-\partial_y^2-H$ on the interval $(y_1,0)$ is positive.

Introduce functions $h_1(y)$ and $h_2(y)$ solving the problem with $f=0$ and with $(\alpha_1,\alpha_2)=(0,1)$ and  $(\alpha_1,\alpha_2)=(1,0)$
respectively. According to the generalized maximum principle (see \cite{PW}),
\begin{eqnarray}\label{F23aa}
&&h_1(y)>0,\;\;h_2(y)>0\;\;\mbox{in $(y_1,0)$ and}\nonumber\\
&&h_{1y}(0)>0,\;\;h_{1y}(y_1)>0,\;\;h_{2y}(0)<0,\;\;h_{2y}(y_1)<0.
\end{eqnarray}
 Now the solution of (\ref{F23a}) is given by
\begin{equation}\label{M1b}
U(y)=\frac{1}{a}\Big(h_1(y)\int_y^0h_2(z)f(z)dz+h_2(y)\int_{-d}^y h_1(z)f(z)dz\Big)
+\alpha_1h_2(y)+\alpha_2h_1(y),
\end{equation}
where
$$
a=h_{2y}h_1-h_2h_{1y}=h_{2y}(0)=-h_{1y}(y_1)<0.
$$
In particular, if $\alpha_1=\alpha_2=0$ then
\begin{equation}\label{M9a}
U'(0)=\int_{-d}^0 h_1(z)f(z)dz.
\end{equation}

\subsection{Dispersion equation}\label{SA14a}

Let a laminar flow $(\Psi(y),d)$ solve the problem (\ref{J18aa})--(\ref{J18ac}) on the interval $(-d,0)$. We are looking for a solution to the problem (\ref{Okt5aa}), (\ref{Ma13bb}), (\ref{Ma13bc}) in the form
$$
\psi(x,y)=\Psi(y)+\varepsilon u(y)\cos (\tau x)+O(\varepsilon^2)
$$
and
$$
\eta(x)=\varepsilon \alpha\cos(\tau x)+O(\varepsilon^2),
$$
where $\varepsilon$ is a small number.
Then $\alpha=-u(0)/\Psi'(0)$ and $u$ together with $\tau$ solves
the spectral problem (\ref{F21a}).
In order to obtain a scalar equation to express solvability of the spectral problem, we introduce the function $\gamma=\gamma (y;\tau)$ as the solution of the problem
\begin{eqnarray}\label{F21abc}
&&(\partial_y^2-\tau^2)\gamma+\omega_p(y,\Psi)\gamma=0\;\;\mbox{in $(-d,0)$},\nonumber\\
&&\gamma(0)=1,\nonumber\\
&& \gamma(-d)=0.
\end{eqnarray}
Due to Assumption (A) this problem is uniquely solvable and by Sect. \ref{SM2a} $\gamma>0$ on $(-d,0]$.
Let $\sigma(\tau;s)=\Psi_y\gamma_y-\kappa |_{y=0}$, where $\kappa$ is defined by (\ref{F21aa}).  Then the spectral problem has a solution $(\tau,u)$ if and only if the dispersion equation (\ref{F21ac}) holds
%\begin{equation}\label{F21ac}
%\sigma(\tau;s)=0
%\end{equation}
and the corresponding eigenfunction in (\ref{F21a}) is $u=\gamma(y;\tau_*)$, where $\tau_*$ denotes the solution of the dispersion equation.

The first simple property is that the function $\sigma(\tau)=\sigma(\tau;s)$ is strongly increasing with respect with $\tau\in [0,\infty)$. To prove this, we differentiate the problem (\ref{F21abc}) with respect to $\tau$ and obtain
\begin{eqnarray*}
&&(\partial_y^2-\tau^2)h+\omega_p(y,\Psi)h=2\gamma\;\;\mbox{in $(-d,0)$},\\
&&h(0)=0,\\
&& h(-d)=0,
\end{eqnarray*}
where $h(y)=h(y;\tau)=\tau^{-1}\gamma_\tau(y;\tau)$.
By Assumption (A) the operator $-\partial_y^2+\tau^2-\omega_p(y,\Psi)$ considered on $(-d,0)$ with the Dirichlet boundary condition is positive definite for $\tau=0$ and hence for all $\tau$. Therefore we can use results of Sect. \ref{SM2a}.
The right hand side of the problem is positive for $\tau>0$. Applying formula (\ref{M1b}), we get $h<0$ for $y\in (-d,0)$ and hence $\gamma_{y\tau}(0;\tau)\geq h_y(0)\tau >0$.
Therefore
\begin{equation}\label{F21af}
\sigma_\tau(\tau)=\gamma_{y\tau}(0;\tau)>0\;\;\mbox{for $\tau>0$ and }\;\sigma(\tau)\to\infty\;\;\mbox{as}\;\tau\to\infty.
\end{equation}
Due to this property  in order to prove the existence of the solution to the dispersion equation (\ref{F21ac}) it is sufficient to investigate the sign of $\sigma(0;s)$. In the next lemma we present a formula for $\sigma(0;s)$.

\begin{lemma} Let $s\neq 0$ and $s<s_*$. Then (\ref{F13ab}) holds,
%\begin{equation}\label{F13ab}
%\sigma(0;s)=\frac{-s-g\rho_0\dot{d}(s)-s\dot{d}V_y(0)}{s\dot{d}(s)},\;\;s=\psi_y(0),
%\end{equation}
where $V$ is the solution of the problem (\ref{F13a}) and
%\begin{eqnarray}\label{F13a}
%&&V^{''}(y)+\omega_p(y,\Psi(y;s))V=\rho_p(-\Psi)\;\;\mbox{on $(-d,0)$}\nonumber\\
%&&V(-d)=V(0)=0.
%\end{eqnarray}

\begin{equation}\label{F18b}
V_y(0)=\int_{-d}^0\gamma(z;0)\rho_p(-\Psi(z;s))d\,z\leq 0
\end{equation}
if the fluid is stably stratified, i.e. $\rho_p\leq 0$.
\end{lemma}

\begin{proof} First let us show that
\begin{equation}\label{F13aa}
\gamma(y;0)=\frac{-\dot{\Psi}(y;s)+\dot{d}(s)(\Psi'(y;s)-V(y))}{s\dot{d}(s)},
\end{equation}
where $V$ is the solution of the problem (\ref{F13a}).

Indeed the relations (\ref{F13a}) guarantee that the function $\gamma $ satisfies the first equation in (\ref{F21abc}). Using that $-\Psi'(-d;s)\dot{d}(s)+\dot{\Psi}(-d;s)=0$, one can check that $\gamma$ satisfies the last relation in (\ref{F21abc}). The second relation in (\ref{F21abc}) is verified directly. This leads to the formula (\ref{F13aa}).

Using (\ref{F13aa}), we write
\begin{eqnarray*}
&&\sigma(0;s)=\frac{-\Psi_{sy}(0;s)+\dot{d}(s)(\Psi_{yy}(0;s)-V_y(0))}{\dot{d}}
-\Psi_{yy}(0;s)-\frac{g\rho_0}{\Psi_y(0;s)}\\
&&=\frac{-s-g\rho_0\dot{d}(s)-s\dot{d}V_y(0)}{s\dot{d}}.
\end{eqnarray*}
Now the formula for $V_y(0)$ follows from (\ref{M9a}).

\end{proof}

A straightforward application of this lemma is the following

\begin{corollary}\label{CM9a} We have
$$
\sigma(0;s)=\frac{-s+O(1)}{s\dot{d}(s)}<0\;\;\mbox{for large negative $s$}.
$$
If the fluid is stably stratified then
$$
\sigma(0;s)=\frac{-\dot{d}(s)(g\rho_0+sV_y(0;s))+O(1)}{s\dot{d}(s)}>0
$$
when $s_*\leq 0$ and $s\to s_*$.
If $s_*>0$, then
$$
\sigma(0;s)=-\frac{g\rho_0}{s}+O(1)>0\;\;\mbox{for small negative $s$}.
$$
%and
%$$
%\sigma(0;s)=-\frac{\rho_0}{s}+O(1)\;\;\mbox{for small positive $s$}
%$$
\end{corollary}

\subsection{Historical remarks}

Laminar flows and dispersion equations for stratified steady Stokes waves were considered in many papers.
In the paper \cite{Wal},  unidirectional laminar flows together with the dispersion equations were considered. The author obtained an estimate for the parameters of the problem which guarantees the existence of roots of the dispersion equation.
Another type of laminar flows were suggested in the paper \cite{W2}, where the vorticity and the density functions  depend on the parameter $d$. Such laminar flows allows to include in the consideration the dispersion equation with small roots, which is important in study of solitary waves.

Here we present a condition on vorticity and the density functions, which guarantees the existence of laminar flows with dispersion equations having small roots. We note that the  laminar flows are not necessary unidirectional in our case.

\section{Bifurcation branches of water waves}

In this section we use the study of the laminar solutions and the dispersion equation
performed in the previous sections. The first step here is a reduction of the water wave problem with unknown boundary to a problem defined on a fixed domain.

\subsection{Flattening the boundary}\label{SOkt13b}

We assume that $\psi_y\neq 0$ only on the boundary $S$
and  use the following flattening change of variables
\begin{equation}\label{M11b}
\hat{x}=\lambda x,\;\;\;\hat{y}=\frac{(y+d)}{\xi(x)+d},
\end{equation}
where
\begin{equation}\label{M11ba}
\lambda=\frac{\Lambda_*}{\Lambda},
\end{equation}
to reduce the problem to a fixed period $\Lambda_*$ and to a  strip with fixed depth
$$
Q=\{(\hat{x},\hat{y})\,:\,\hat{x}\in \Bbb R,\;\;0<\hat{y}<1\}.
$$
Since
$$
\partial_x=\lambda\Big(\partial_{\hat{x}}-\frac{\hat{y}\eta'}{\eta+d}\partial_{\hat{y}}
\Big),\;\;\partial_y=\frac{1}{\eta+d}\partial_{\hat{y}},
$$
where
$$
\eta(\hat{x})=\xi (\lambda^{-1}\hat{x})
$$
and $'$ means $\partial_{\hat{x}}$,
the equations (\ref{Okt5aa}) and  (\ref{Ma13bc}) takes the form
\begin{eqnarray}\label{K2aa}
&&F(\hat{\psi},\eta;\lambda):=\Big(\lambda^2\Big(\partial_{\hat{x}}-\frac{\hat{y}\eta'}{\eta+d}
\partial_{\hat{y}}\Big)^2+\Big(\frac{1}{\eta+d}\partial_{\hat{y}}\Big)^2\Big)
\hat{\psi}+\hat{\omega}(\hat{y},\hat{\psi})=0\;\;\mbox{in $Q$},\nonumber\\
&&G(\hat{\psi},\xi,\lambda):=\frac{1}{2}\Big(\lambda^2\Big|\Big(\partial_{\hat{x}}-
\frac{\hat{y}\eta'}{\eta+d}\partial_{\hat{y}}\Big)\hat{\psi}\Big|^2
+\Big|\frac{1}{\eta+d}\partial_{\hat{y}}\hat{\psi}\Big|^2\Big)\nonumber\\
&&+g\rho(0)(\eta(\hat{x})+d)-R=0\;\;\mbox{for $\hat{y}=1$},\nonumber\\
&&\hat{\psi}=0\;\;\mbox{for $\hat{y}=1$},\nonumber\\
&&\hat{\psi}=-p_0\;\;\mbox{for $\hat{y}=0$},
\end{eqnarray}
where
$$
\hat{\psi}(\hat{x},\hat{y})=\psi\Big(\lambda^{-1}\hat{x},
\hat{y}(\eta(\hat{x})+d)-d\Big)
\;\;\mbox{or}\;\;\psi(x,y)=\hat{\psi}(\hat{x},\hat{y})
$$
and
$$
\hat{\omega}(\hat{y},\hat{\psi})=\omega(y,\psi(x,y))=\omega\Big(
\hat{y}(\eta(\hat{x})+d)-d,\hat{\psi}(\hat{x},\hat{y})\Big).
$$
Then the problem (\ref{K2aa}) is equivalent to
\begin{equation}\label{Ma31aa}
(F(\hat{\psi},\eta,\lambda),G(\hat{\psi},\eta,\lambda))=0,
\end{equation}
which is defined on $\Lambda_*$--periodic, even functions from
$C^{2,\alpha}(Q)\times C^{2,\alpha}(\Bbb R)$  satisfying $\hat{\psi}(\hat{x},0)=-p_0$,
$\hat{\psi}(\hat{x},1)=0$ and $\eta+d>0$.
Here and in what follows we use the spaces $C^{k,\alpha}_{\textrm{b}}(\overline{Q})$, $C^{k,\alpha}_{\Lambda_*}(\overline{Q})$,
$C^{k,\alpha}_{0,e,\Lambda_*}(\overline{Q})$ which are defined similar to the spaces in $D$ introduced before Theorem \ref{TA17b}.

%Let us introduce some functional spaces.
%For $k=0,1,\ldots$, $\alpha\in(0,1)$ and $-\infty<a<b<\infty$
%we denote by $C^{k,\alpha}(\overline{Q}_{a,b})$ and $C^{k,\alpha}([a,b])$
%the H\"older spaces of functions in $\overline{Q}_{a,b}=\{(x,y)\in \overline{Q}
%\,:\,a\leq x\leq b\}$ and $[a,b]$ respectively.

%The space $C^{k,\alpha}_{\textrm{b}}(\overline{Q})$ consists of functions
%$u$ defined on $\overline{Q}$
%such that
%$$
%\sup_{a\in \Bbb R}||w||_{C^{k,\alpha}(\overline{Q}_{a,a+1})}<\infty.
%$$
% The subspaces   $C^{k,\alpha}_{\Lambda_*}(\overline{Q})$,
%$C^{k,\alpha}_{0,e,\Lambda_*}(\overline{Q})$  consists of
%$\Lambda_*$-periodic ($\Lambda_*$--periodic, even) functions in
%$C^{k,\alpha}_{\textrm{b}}(\overline{Q})$ vanishing for $y=-d$.

%Similarly, one can define the spaces $C^{k,\alpha}_{\textrm{b}}(\Bbb R)$,
%$C^{k,\alpha}_{\Lambda_*}(\Bbb R)$ and
%$C^{k,\alpha}_{e,\Lambda_*}(\Bbb R)$.

We remind that it is assumed that the density function $\rho$ is of class $C^{2,\alpha}$
for a certain $\alpha\in (0,1)$ and the Bernoulli function $\beta$ is of class
$C^{1,\alpha}$.

The introduced above operators are continuous in the following spaces
\begin{eqnarray*}
&&(F(\hat{\psi},\eta,\lambda),G(\hat{\psi},\eta,\lambda))\;:\;
C^{2,\alpha}_{0,e,\Lambda_*}(\overline{Q})\times C^{2,\alpha}_{e,\Lambda_*}(\Bbb R)\times\Bbb R_+\\
&&\rightarrow C^{0,\alpha}_{0,e,\Lambda_*}(\overline{Q})\times C^{1,\alpha}_{e,\Lambda_*}(\Bbb R).
\end{eqnarray*}

\begin{theorem}\label{TA17a} We assume that the laminar flow $(\Psi_*,d_*)$ is chosen such that the dispersion  equation (\ref{F21ac}) has a solution $\tau_*>0$ with the corresponding eigenfunction $\gamma(y;\tau_*)$.
Then there exist $\varepsilon > 0$ and a mapping
\begin{equation}\label{Ma2bb}
(-\varepsilon,\varepsilon)\ni t\rightarrow (\hat{\psi}[t],\eta[t],\lambda(t)))\in C^{2,\alpha}_{0,e,\Lambda_*}(\overline{Q})\times C^{2,\alpha}_{e,\Lambda_*}(\Bbb R)\times\Bbb R_+,
\end{equation}
where $(\hat{\psi}[t],\eta[t],\lambda(t))$ is an even, $\Lambda_*$--periodic solution to (\ref{K2aa}). The mapping is continuous together with the first derivatives with respect to $t$. Furthermore,
\begin{eqnarray}\label{A17aczz}
&&\lambda(0)=1,\;\;\eta[t]=-\frac{t}{\Psi_{*y}(0)}\cos(\tau_*\hat{x})+o(|t|)
\;\;\mbox{and}\nonumber\\
&&\psi[t]=\Psi_*(y)+t\gamma(y;\tau_*)\cos(\tau_*\hat{x})+o(|t|).
\end{eqnarray}

\end{theorem}

Clearly, Theorem \ref{TA17b} follows from this assertion.
Remark \ref{RA25a} is a consequence of the next
\begin{remark}\label{RA25az}
1). If the functions $\rho$ and $\beta$ are real analytic then the functions $\lambda(t)$, $\psi[t]$ and $\eta[t]$ are also real analytic. This can be prove by complementing
the Crandall--Rabinowitz theorem by results from Chapter I.16, \cite{Ki1}.

2). The smoothness property of the function $\psi[t]$ mentioned in Remark \ref{RA25a} is equivalent to the smoothness property of the function $\widehat{\psi}$ formulated in the above theorem.
\end{remark}

Theorem \ref{TMa2a} in the new variables has the following form.

\begin{theorem}\label{TMa2b} We assume that all assumptions from Theorem \ref{TA17a} are fulfilled. Moreover the functions $\rho$ and $\beta$ are real analytic. Then there exists the branch
\begin{equation}\label{Ma2ba}
\Bbb R\ni t\rightarrow (\hat{\psi}[t],\eta[t],\lambda(t)))\in C^{2,\alpha}_{0,e,\Lambda_*}(\overline{Q})\times C^{2,\alpha}_{e,\Lambda_*}(\Bbb R)\times\Bbb R_+,
\end{equation}
where $(\hat{\psi}[t],\eta[t],\lambda(t))$ is a even, $\Lambda_*$--periodic solution to (\ref{K2aa})  coinciding with the branch (\ref{Ma2bb}) for small $t$. The vector function (\ref{Ma2ba}) is real analytic with respect to $t$ (up to a local reparameterization).

Furthermore, the second part of Theorem \ref{TMa2a} concerning the alternatives (a) or (b) holds.
% there exists a sequence $(\psi[t_j],\eta[t_j],\lambda_j)$, $j=1,2,\ldots$, such %that (\ref{Jun7a}) is valid and
%one of the alternatives (a) or (b) of Theorem \ref{TMa2a} hold.

%(a) there is $\varepsilon_*>0$ and  such all element of the branch belongs to %$\widehat{\mathcal U}_{\varepsilon_*}$ and infinity many terms of the sequence %$(\psi[t_j],\xi[t_j],\Lambda_j)$ lie outside
%${\mathcal U}_{\delta,\varepsilon_*}$ for each $\delta>0$;

%(b) $\varepsilon(\psi[t_j],\xi[t_j],\Lambda_j)\to 0$ as $j\to\infty$.

\end{theorem}

\subsection{Laminar solutions}

Laminar solutions of the problem (\ref{K2aa}) are solutions which depend only on $\hat{y}$, when $\eta=0$ and $\hat{\psi}$ and $d$ are considered as unknowns. Then they satisfy
\begin{eqnarray}\label{K2aaa}
&&\frac{1}{d^2}\partial_{\hat{y}}^2
\hat{\psi}+\hat{\omega}(\hat{y},\hat{\psi})=0\;\;\mbox{in $Q$},\nonumber\\
&&\frac{1}{2d^2}\big|\partial_{\hat{y}}\hat{\psi}\big|^2+g\rho(0)d-R=0
\;\;\mbox{for $\hat{y}=d$},\nonumber\\
&&\hat{\psi}=0\;\;\mbox{for $\hat{y}=1$},\nonumber\\
&&\hat{\psi}=-p_0\;\;\mbox{for $\hat{y}=0$}.
\end{eqnarray}
After the change of variable $y=d(\hat{y}-1)$ this problem transforms to the problem
(\ref{J18aa})--(\ref{J18ac}).

\subsection{Fr\'{e}chet derivative}\label{SA15a}

We calculate the Fr\'{e}chet derivative at $(\hat{\psi},\eta)$:
\begin{eqnarray}\label{Ju28a}
&&{\mathcal F}(u,\zeta,\lambda):=\partial_tF(\hat{\psi}+tu,\eta+t\zeta)|_{t=0}\nonumber\\
&&=\Big(\lambda^2\Big(\partial_{\hat{x}}-\frac{\hat{y}\eta'}{\eta+d}
\partial_{\hat{y}}\Big)^2+\Big(\frac{1}{\eta+d}\partial_{\hat{y}}\Big)^2\Big)
u+\tilde{\omega}
-\lambda^2\Big(\frac{\zeta}{\eta+d}\Big)'\hat{y}\partial_{\hat{y}}
\Big(\partial_{\hat{x}}
-\frac{\eta'}{\eta+d}{\hat{y}}
\partial_{\hat{y}}\Big)\hat{\psi}\nonumber\\
&&-\lambda^2\Big(\partial_{\hat{x}}-\frac{\hat{y}\eta'}{\eta+d}\partial_{\hat{y}}\Big)
\Big(\frac{\zeta}{\eta+d}\Big)'\hat{y}\partial_{\hat{y}}\hat{\psi}
-2\frac{\zeta}{(\eta+d)^3}\partial_{\hat{y}}^2\hat{\psi},
\end{eqnarray}
where
$$
\tilde{\omega}=\omega_\psi(y,\psi)u+\omega_y\hat{y}\zeta,
$$
and
\begin{eqnarray}\label{Ju28aa}
&&{\mathcal G}(u,\zeta,\lambda):=\partial_tG(\hat{\psi}+tu,\xi+t\zeta)|_{t=0}=
\lambda^2\Big(\partial_{\hat{x}}\hat{\psi}-\frac{\hat{y}\eta'}{\eta+d}\partial_{\hat{y}}
\hat{\psi}\Big)\Big(\partial_{\hat{x}}u
-\frac{\hat{y}\eta'}{\eta+d}\partial_{\hat{y}}u\Big)\nonumber\\
&&\!\!\!+\!\!\frac{1}{(\eta+d)^2}\partial_{\hat{y}}\hat{\psi}\partial_{\hat{y}}u
\!\!+\!\!g\rho(0)\zeta
\!\!-\!\!\lambda^2\Big(\partial_{\hat{x}}\hat{\psi}-\frac{\hat{y}\eta'}{\eta+d}\partial_{\hat{y}}
\hat{\psi}\Big)\Big(\frac{\zeta}{\eta+d}\Big)'\hat{y}\partial_{\hat{y}}\hat{\psi}
\!\!-\!\frac{\zeta}{(\eta+d)^3}\hat{\psi}_{\hat{y}}^2.
\end{eqnarray}
Here $u=0$ for $\hat{y}=0$ and $\hat{y}=1$. We introduce
\begin{equation}\label{A16a}
\omega_*=\omega_*(y,\psi)=\omega_\psi(y,\psi)=-\beta'(\psi)+gy\rho^{''}(-\psi).
\end{equation}
Then
\begin{equation}\label{A16aa}
\tilde{\omega}=\omega_*(y,\psi)u+\omega_y\hat{y}\zeta\;\;\mbox{and}\;\;\omega_y=-g\rho'(-\psi).
\end{equation}

Fr\'{e}chet derivative at the laminar solution $(\Psi,d)$ is
\begin{equation}\label{Ju28az}
{\mathcal F}(u,\zeta,\lambda)=
\Big(\lambda^2\partial_{\hat{x}}^2+\frac{1}{d^2}\partial_{\hat{y}}^2\Big)
u+\tilde{\omega}
-\lambda^2\frac{\zeta^{''}}{d}\hat{y}\partial_{\hat{y}}\hat{\Psi}
-2\frac{\zeta}{d^3}\partial_{\hat{y}}^2\hat{\Psi}
\end{equation}
and
\begin{equation}\label{Ju28aaz}
{\mathcal G}(u,\zeta,\lambda)=
\frac{1}{d^2}\partial_{\hat{y}}\hat{\Psi}\partial_{\hat{y}}u
+g\rho(0)\zeta
-\frac{\zeta}{d^3}\hat{\Psi}_{\hat{y}}^2.
\end{equation}

Let us introduce the transformation
\begin{equation}\label{Au2a}
v(x,y)=u(\hat{x},\hat{y})-\Psi_{y}(y)\frac{(y+d)\zeta}{d}.
\end{equation}
Since $u=0$ for $\hat{y}=0$ and $\hat{y}=1$,
\begin{equation}\label{A14a}
v(x,-d)=0\;\;\mbox{and}\;\;v(x,0)=-\Psi_y(0)\zeta(x).
\end{equation}

\begin{lemma}\label{LA17a} (i) Assume that the pair $(\Psi,d)$  satisfy (\ref{J18aa}). If the function  $v$ is given by (\ref{Au2a}) then
\begin{equation}\label{A14aa}
(\partial_x^2+\partial_y^2)v+\omega_*v={\mathcal F}(u,\zeta),
\end{equation}
where $\omega_*$ is defined by (\ref{A16a}).

(ii) Furthermore
\begin{equation}\label{A14ab}
\Psi_yv_y+\hat{\sigma}\zeta={\mathcal G}(u,\zeta)\;\;
\mbox{on ${\mathcal S}_\xi$},
\end{equation}
where
$$
\hat{\sigma}=\psi_y\psi_{yy}+g\rho(0).
$$

\end{lemma}
\begin{proof} (i) Using relations (\ref{Ju28a}) and (\ref{Okt5aa}), we get
\begin{eqnarray*}
&&(\partial_x^2+\partial_y^2)v+\omega_*v=\Big(\partial_{\hat{x}}^2
+\frac{1}{d^2}\partial_{\hat{y}}^2 \Big)u(\hat{x},\hat{y})+\tilde{\omega}\\
&&-\Big(\partial_x^2+\partial_y^2+\omega_*\Big)\Big(\Psi_{y}(y)
\frac{(y+d)\zeta}{d}\Big)-\frac{y+d}{d}\omega_y\zeta\\
&&=\Big(\partial_{\hat{x}}^2
+\frac{1}{d^2}\partial_{\hat{y}}^2 \Big)u(\hat{x},\hat{y})+\tilde{\omega}
\\
&&-\frac{(y+d)\zeta}{d}\Big(\partial_y^2\Psi_{y}(y)
+\omega_*\Psi_{y}(y)+\omega_y\Big)
-\frac{y+d}{d}\Psi_y\zeta^{''}-2\Psi_{yy}\frac{\zeta}{d}.
\end{eqnarray*}
 Comparing this with the second line in (\ref{Ju28a}) and using that $(\partial_y^2
+\omega_*)\Psi_{y}-g\rho'(-\Psi)=0$, we arrive at the assertion (i).

(ii) We have
\begin{eqnarray*}
&&\Psi_yv_y+\hat{\sigma}\zeta=\psi_yu_y+\hat{\sigma}\zeta
-(\Psi_y\Psi_{yy})\zeta-\Psi_y^2\frac{\zeta}{d}\\
&&=\Psi_yu_y+g\rho(0)\zeta-\Psi_y^2\frac{\zeta}{d}.
\end{eqnarray*}
This together with (\ref{Ju28aa}) leads to the required proof of (ii).
\end{proof}

The element of the  kernel of the Fr\'{e}chet derivative in $\hat{x}$, $\hat{y}$ variables satisfy
$$
{\mathcal F}(u,\zeta;\lambda)=0,\;\;{\mathcal G}(u,\zeta;\lambda)=0,
$$
where $u=0$ for $\hat{y}=0$ and $\hat{y}=1$. According to the previous lemma we have
\begin{eqnarray}\label{A15aa}
&&(\partial_x^2+\partial_y^2)v+\omega_*v=0
\;\;\mbox{in $D_\xi$},\nonumber\\
&&\Psi_yv_y+\hat{\sigma}\zeta=0\;\;
\mbox{for $y=0$},\nonumber\\
&&v(x,-d)=0\;\;\mbox{and}\;\;v(x,0)=-\Psi_y(0)\zeta(x).
\end{eqnarray}

Let us introduce the function $\gamma(y;\tau)$  solving the problems
\begin{eqnarray}\label{A14b}
&&(-\tau^2+\partial_y^2)\gamma+\omega_*\gamma=0
\;\;\mbox{for $y\in (-d,0)$}\nonumber\\
&&\gamma(-d;\tau)=0,\;\; \gamma(0;\tau)=1.
\end{eqnarray}
If we are looking for solutions of the  problem in the form
\begin{equation}\label{A15ab}
v(x,y)=\gamma(y;\tau)\cos(\tau x)\;\;\mbox{and}\;\;\zeta(x)=\alpha\cos(\tau x),
\end{equation}
then
 we obtain the following boundary condition for $\gamma$
\begin{equation}\label{A14bc}
\Psi_y\gamma_y-\frac{\hat{\sigma}}{\Psi_y}=0\;\;\mbox{for $y=0$}.
\end{equation}
Here we have used also the second relation in (\ref{A14a}). The relation (\ref{A14bc}) represents the dispersion equation for finding the frequency $\tau$. The corresponding element of the kernel is given by (\ref{A15ab}), where
$\alpha=-1/\Psi_y(0)$.

\subsection{Small amplitude water waves. Proof of Theorem \ref{TA17a}}

Our proof is based on the classical Crandall--Rabinowitz Theorem. It deals with the following problem. Let ${\mathcal X}$, ${\mathcal Y}$ be Banach spaces and  $I$ be a finite or infinite interval in $\Bbb R$. Let also ${\mathcal U}$ be an open subset in ${\mathcal X}\times I$ and
\begin{equation}\label{Ma12a}
\widehat{\mathcal F}\; :\; {\mathcal U} \rightarrow {\mathcal Y}
\end{equation}
be a continuous function. We are interested in solutions of the equation
\begin{equation}\label{Ma31b}
\widehat{\mathcal F}(\Phi,\lambda)=0,\;\;(\Phi,\lambda)\in {\mathcal U}
\end{equation}
such that the derivatives
\begin{equation}\label{Ma12aa}
\partial_\Phi \widehat{\mathcal F},\;\;\partial_{\Phi\Phi}\widehat{\mathcal F},\;\; \partial_\lambda\widehat{\mathcal F}\;\;\mbox{ and}\;\; \partial_{\lambda\Phi}\widehat{\mathcal F}
\end{equation}
 are also continuous.

\begin{theorem}\label{TA17aaz}(Crandall--Rabinowitz, \cite{CR}) Let the derivatives of the operator $\widehat{\mathcal F}$:
\begin{equation}\label{Ma12aaz}
\partial_\Phi \widehat{\mathcal F},\;\;\partial_{\Phi\Phi}\widehat{\mathcal F},\;\; \partial_\lambda\widehat{\mathcal F}\;\;\mbox{ and}\;\; \partial_{\lambda\Phi}\widehat{\mathcal F}
\end{equation}
 are continuous. Suppose that

(i)  $\widehat{\mathcal F}(0,\lambda) = 0$ holds for all $(0,\lambda)\in {\mathcal U}$,

(ii) For some $\lambda_*\in I$ such that $(0,\lambda_*)\in {\mathcal U}$ the operator
$\partial_\Phi\widehat{\mathcal F}(0,\lambda_*)$ is a Fredholm operator with zero index and the null-space of $\widehat{\mathcal F}_\Phi(0,\lambda_*)$ is
one-dimensional and generated by $\Phi^0$ and
\begin{equation}\label{Ma12b}
\widehat{\mathcal F}_{\Phi,\lambda}(0,\lambda_*)\Phi^0\;\;\mbox{ does
not belong to the range of}\;\; \widehat{\mathcal F}_\Phi(0,\lambda_*).
\end{equation}

If (i) and (ii) hold, then a sufficiently small $\varepsilon > 0$ exists and a $C^1$ curve of solutions to (\ref{Ma31b}):
$$
\{(\Phi(t),\lambda(t)) : |t| <\varepsilon\} \subset  {\mathcal X}\times I,
$$
which bifurcates from $(0,\lambda_*)$. Moreover, for pairs belonging to this curve the following properties hold:
$$
\Phi(t) = t\Phi^0 + o(t)\;\;\mbox{ when $0 < |t| < \varepsilon$},
$$
and $\lambda(0)=\lambda_*$,
$$
\{(\lambda,\Phi) \in V : \lambda\neq \lambda_*\;\mbox{ and}\; \widehat{\mathcal F}(\lambda,\Phi) = 0\} = \{(\lambda(t),\Phi(t)) : 0 < |t| < \varepsilon\},
$$
where
$V \subset {\mathcal X} \times I$ is a certain neighbourhood of $(0,\lambda_*)$.
\end{theorem}

For application of this theorem we choose $d=d_*$, $\lambda_*=1$, $I=\Bbb R_+$ and
$$
{\mathcal X}=\{(\Phi=(w,\zeta)\;:\; w\in C^{2,\alpha}_{0,e,\Lambda_*}(\overline{Q}), \;w(\cdot,1)=0,\;\zeta\in C^{2,\alpha}_{e,\Lambda_*}(\Bbb R)\},
$$
$$
{\mathcal Y}=C^{\alpha}_{e,\Lambda_*}(\overline{Q})\times C^{1,\alpha}_{e,\Lambda_*}(\Bbb R).
$$
Let ${\mathcal U}$ be the set of functions $\Phi=(w,\zeta)$ in ${\mathcal X}$, given by (\ref{A30b}).

We introduce
$$
\widehat{\mathcal F}(\Phi,\lambda)=(F(\Psi_*+w,\zeta,\lambda),G(\Psi_*+w,\zeta,\lambda)).
$$
 It is clear that the operator
 \begin{equation}\label{Ma31ab}
 \widehat{\mathcal F}(\Phi,\lambda)\;:\;{\mathcal U}\rightarrow {\mathcal Y}
 \end{equation}
is continuous together with derivatives (\ref{Ma12aa}) and $\widehat{\mathcal F}(0,\lambda)=0$.

It remains to verify the condition (ii) of the above theorem.

According to Sect. \ref{SA15a} the Fr\'{e}chet derivative of the operator (\ref{Ma31aa}) at a laminar solution $(\Psi_*,d_*)$ in the variables $(x,y)$ is given by Lemma \ref{Ju28az}.
Therefore the condition for existence of a positive $\tau$ is given by equation (\ref{A14bc}), which have a solution for a certain parameter $s$, when for example $\rho'(-\Psi)\leq 0$. This solution is unique and it was denoted by $\tau_*$ in Theorem \ref{TA17b}. The corresponding eigenfunction is $u_*(y)=\gamma(y;\tau_*)\cos(\tau_*\hat{x})$.

 Simplicity of the eigenvalue follows from the simplicity of the eigenvalue of the problem (\ref{F21a}) and monotonicity of the function $\sigma(\tau)$.
 The Fr\'{e}chet derivative is a Fredholm operator since it is a self-adjoint operator in $(x,y)$ variables.

 In $\hat{x}$, $y$ variables the Fr\'{e}chet derivative has the form
\begin{eqnarray}\label{A17ca}
&&(\partial_y^2-\lambda^2\tau^2)u+\omega_p(y,\Psi)u
\;\;\mbox{in $(-d,0)$},\nonumber\\
&&\Psi_yu_y-\kappa u\;\;\mbox{for $y=0$},\nonumber\\
&& u(-d)=0.
\end{eqnarray}
Therefore,
$$
\widehat{\mathcal F}_{\Phi,\lambda}(0,\lambda_*)\Phi^0=-2\tau_*^2u_*
$$
and to prove (\ref{Ma12b}), we must show that $-2\tau_*^2\int_{-d_*}^0u_*u_*dy\neq 0$, which is certainly true.

Thus, all assumptions in the Crandall--Rabinowitz Theorem are verified and its application leads to the proof of Theorem \ref{TA17a}.

\subsection{Water waves of large amplitude}

Here we use the same notations as in Theorem \ref{TA17aaz}. %The set ${\mathcal U}$ is given by (\ref{A30b}) and
Introduce
\begin{equation}\label{D5b}
{\mathcal S}=\{(\Phi,\lambda)\in{\mathcal U}\,:\, {\mathcal F}(\Phi,\lambda)=0\}.
\end{equation}
Let
\begin{equation}\label{Ma12ba}
{\mathcal B}_\epsilon=\{(\Phi(t),\lambda(t))\,:\,t\in (-\epsilon,\epsilon)\}
\end{equation}
be the bifurcating branch of small amplitude Stokes waves constructed in Theorem \ref{TA17aaz}.

We use the following version of the global bifurcation theorem taken from \cite{CSrVar}, which is a corrected and modified version of the bifurcation theorem from \cite{BTT}.

\begin{theorem}\label{ThJ3a} Suppose that all conditions of Theorem \ref{TA17aaz} hold and for some sequence ${\mathcal K}_\delta$, $\delta\in (0,\delta_*)$, of bounded closed subsets of ${\mathcal U}$ with ${\mathcal U}=\bigcup_{\delta\in (0,\delta_*)}{\mathcal K}_\delta$, the set ${\mathcal S}\cap {\mathcal K}_\delta$ is compact for each $\delta\in (0,\delta_*)$ and the Frech\'{e}t derivatives of the operator ${\mathcal F}$ is a Fredholm operator of index zero at all point from ${\mathcal S}$.
Then there exists a continuous curve- ${\mathcal B}$, which extends ${\mathcal B}_\epsilon$ as follows

\medskip
(a) ${\mathcal B}=\{(\Phi(t),\lambda(t))\,:\, t\in \Bbb R\}\subset {\mathcal U}$,
where $(\Phi,\lambda): \Bbb R\rightarrow {\bf X}\times \Bbb R_+$ is continuous;

\medskip
(b) ${\mathcal B}_\varepsilon\subset {\mathcal B}\subset {\mathcal S}$;

\medskip
(c) ${\mathcal B}$ has a real-analytic reparametrization locally around each of its points;

\medskip
(d) One of the following alternatives occurs:

$(\alpha)$ for every  $\delta\in (0,\delta_*)$, there exists $t_\delta>0$ such that $(\Phi(t);\lambda(t))$ does not belong to ${\mathcal K}_\delta$ for all $t\in\Bbb R$ with $|t|>t_\delta$;

$(\beta)$ there exists $T>0$ such that $(\Phi(t+T),\lambda(t+T))=(\Phi(t),\lambda(t))$ for all $t>0$.

Moreover, such a curve of solutions to $\widehat{\mathcal F}(\Phi,\lambda)=0$ having the properties (a)-(d) is unique (up to reparametrization).

\end{theorem}

In order to apply this theorem, we choose
\begin{equation}\label{Jun16a}
{\mathcal K}_\delta =\overline{{\mathcal U}_\delta},
\end{equation}
where ${\mathcal U}_\delta$ is defined by (\ref{A30ac}). Then the formula 
\begin{equation}\label{Jun16aa}
{\mathcal U}=\bigcup_\delta \overline{{\mathcal U}_\delta}
\end{equation}
holds because of (\ref{A30b}) and monotonicity of ${\mathcal U}_\delta$ with respect to $\delta$.
 We note also that ${\mathcal S}$ coincides with ${\mathcal O}$ given by (\ref{A30ba}).%We take ${\mathcal U}$ defined by (\ref{A30b}). Then ${\mathcal S}$ coincides with ${\mathcal O}$ given by (\ref{A30ba}).

Let
\begin{equation}\label{A29a}
M=\max_{(x,y)\in \overline{D}}\psi(x,y),
\end{equation}
\begin{equation}\label{Ma27a}
m=\min_{x\in \Bbb R}(\xi(x)+d)
\end{equation}
and
\begin{equation}\label{A29aa}
L=\max_x|\xi'(x)|.
\end{equation}
It is clear that $M\geq |p_0|$.

The following lemma will be important in proving of compactness property, which is one of the  assumptions in the above theorem.
\begin{lemma}\label{LA29a}
There exists $C_*=C_*(C_1,L,M,m)$ depending on the constant $C_1$ from the inequalities (\ref{F11a})   and the constants $L$, $M$ and $m$, such that if a $\Lambda$-- periodic function $(\psi,\xi)$ solves the problem (\ref{J18aa})--(\ref{J18ac}), then the inequality
\begin{equation}\label{Ma27aa}
\psi_x^2(x,y)+\psi_y^2(x,y)\leq C_*
\end{equation}
 holds for all $(x,y)\in \overline{D}$.
\end{lemma}
The proof of this assertion repeats the proof of Proposition 2 in \cite{KLN17}. The only difference is that
we can not estimate $m$ from below by a constant depending only on $R$, $C_1$, $L$ and $M$. We overcome this by including the dependence of $C_*$ on $m$ in the final estimate.
\begin{proof}
Denote
$$
\Omega_s=\{(x,y)\in D_{\xi}\,:\,s<x<s+1\}.
$$
%$$
%\Omega=\{(x,y)\in D_{\xi}\,:\,-\Lambda/2<x<\Lambda/2\}.
%$$
%Then
%\begin{eqnarray*}
%&&0=\int_{\Omega} (\psi+p_0)(\Delta \psi+\omega(\psi) dxdy=-\int_{\Omega} (|\nabla\psi|^2-\omega(y,\psi)(\psi+p_0))dxdy\\
%&&+\int_{-\Lambda/2}^{\Lambda/2}(-\xi'(x)\psi_x+\psi_y) (\psi+p_0)dx.
%\end{eqnarray*}
%Therefore
%$$
%\int_{\Omega} |\nabla\psi|^2 dxdy=\int_{\Omega} \omega(y,\psi)(\psi+p_0))dxdy+\int_{-\Lambda/2}^{\Lambda/2}(-\xi'(x)\psi_x+\psi_y)dx,
%$$
%where we have used the boundary conditions for $\psi$ and the condition (\ref{F11a}) on $\omega$.
%This implies
Repeating the proof of (4.12) in Proposition 2 in \cite{KLN17}, we get
$$
\int_{\Omega_s} |\nabla\psi|^2 dxdy\leq C(R,C_1,M,L),
$$
where $C$ does not depend on $s$. 

Differentiating the equation (\ref{Okt5aa}) with respect to $x$, we obtain
\begin{equation*}
\Delta\psi_x+\omega_p(y,\psi)\psi_x=0\;\;\mbox{in $D$}.
\end{equation*}
Applying Theorem 8.25 in \cite{GT}, we arrive at the estimate (\ref{Ma27aa}) for the function $\psi_x$.
Differentiating the equation (\ref{Okt5aa}) with respect to $y$, we obtain
\begin{equation*}
\Delta\psi_y+\omega_p(y,\psi)\psi_y=g\rho'(\psi)\;\;\mbox{in $D$}.
\end{equation*}
The value of $\psi$ on the bottom $y=-d$ is estimated already by a constant depending on $C_1$, $M$, $m$ and $p_0$.
The pointwise estimate of $\psi_y$ on the bottom can be obtained by applying standard local estimates near the bottom for the Laplace operator. It is estimated by a constant depending on $C_1$, $M$, $m$ and $p_0$.
The value of $\psi_y$ on the free surface $y=\xi$ is estimated by using the boundary condition (\ref{Ma13bc}). Now application of Theorem 8.25 in \cite{GT}
leads to the estimate (\ref{Ma27aa}) for the function $\psi_y$.

\end{proof}

\bigskip
{\bf Proof of Theorem \ref{ThJ3a}}. In order to satisfy the assumptions of the theorem, we will check that

(i) the Frechet derivative at each point $(w,\eta;\lambda)$ in ${\mathcal O}$ is a Fredholm operator of index zero;

(ii) Let
$$
{\mathcal K}_{\delta,\varepsilon}=\overline{{\mathcal U}_{\delta,\varepsilon}},
$$
where ${\mathcal U}_{\delta,\varepsilon}$ is defined by (\ref{Ma30a}).
The set ${\mathcal K}_{\delta,\varepsilon}\cap {\mathcal O}$ is compact in ${\mathcal X}$;

(iii) the alternative $(\beta)$ in (d) does not hold.

\medskip
Let us prove (i). We shall use the following lemma, proved in \cite{KPOMI}, where we applied a slightly different change of variables. Let
$$
\hat{x}=x,\;\;\;\hat{y}=\frac{(y+d)}{\xi(x)+d}
$$
and
let us introduce the transformation
\begin{equation}\label{Au2az}
v(x,y)=u(\hat{x},\hat{y})-\psi_{y}(x,y)\frac{(y+d)\zeta}{\xi+d}.
\end{equation}

\begin{lemma} (i) Assume that the functions $\psi$ and $\xi$ satisfy (\ref{Okt5aa})
in the domain $D_\xi$. If the function  $v$ is given
by (\ref{Au2az}) then
$$
(\partial_x^2+\partial_y^2)v+\omega_*v={\mathcal F}(u,\zeta),
$$
where $\omega_*$ is defined by (\ref{A16a}).
%$$
%\Psi(X,Y)=\psi\Big(X,\frac{dY}{\xi}\Big).
%$$

(ii) Furthermore
$$
\psi_xv_x+\psi_yv_y+\hat{\sigma}\zeta={\mathcal G}(u,\zeta)\;\;
\mbox{on ${\mathcal S}_\xi$},
$$
where
$$
\hat{\sigma}=\psi_x\psi_{xy}+\psi_y\psi_{yy}+g\rho(0),\;\;v=-\psi\,\zeta.
$$

\end{lemma}

By this lemma, the Fredholm property of the operator ${\mathcal F}$  follows from similar property for the operator
\begin{eqnarray}\label{Se28b}
&&A v:=(\partial_x^2+\partial_y^2) v+gy\rho^{''}(-\psi)v-\beta' (\psi)v\;\;
\mbox{in $D$},\nonumber\\
&&B v:=\psi_xv_x+\psi_yv_y-\frac{\hat{\sigma}}{\psi}v\;\;
\mbox{for $y=\xi(x)$}
\end{eqnarray}
defined on functions subject to
\begin{equation}\label{Se28ad}
v(x,-d)=0.
\end{equation}
Since this operator is self-adjoint, we obtain the required property.

\medskip
(ii) In order to prove the compactness property, we need some lemmas. Let $\psi\in\overline{{\mathcal U}_{\delta,\varepsilon}}$. We apply the transformation
$$
q=x,\;\;p=-\psi(x,y),\;\;(x,y)\in \{(x,y)\,:\,x\in\Bbb R,\;y_*(x)\leq y\leq \xi(x)\},
$$
where $y_*(x)$ was introduced in the definition \ref{Ma30a}.
Then
$$
(q,p)\in\overline{\hat{Q}},\;\;\hat{Q}=\{q\in\Bbb R,\;p\in (-\varepsilon,0)\}.
$$
We put
$$
h(q,p)=\varepsilon+y.
$$
Then
$$
\psi_y=-\frac{1}{h_p},\;\;\psi_x=\frac{h_q}{h_p}.
$$
Applying this change of variables, we arrive at
\begin{eqnarray}\label{J17aa}
&&F(h):=\Big(\frac{1+h_q^2}{2h_p^2}\Big)_p-\Big(\frac{h_q}{h_p}\Big)_q
+g(h-\varepsilon)\rho_p
-\beta(-p)=0\;\;
\mbox{in $\hat{Q}$},\nonumber\\
&&G(h):=\frac{1+h_q^2}{2h_p^2}+gh\rho-R=0\;\;\mbox{for $p=0$},\nonumber\\
&&h(q,p_0)=0\;\;q\in\Bbb R.
\end{eqnarray}
In the next lemma we use the notation $B_\rho(q,p)$ for the disc with the center at $(q,p)$ and with the radius $\rho$.

\begin{lemma}\label{LMa31a} Let $h\in C^{2,\alpha}(\overline{\hat{Q}})$  be a solution to (\ref{J17aa}) and $\delta$ and $\varepsilon$ satisfy (\ref{Ma30a}).
Suppose that the corresponding $(\psi,\eta)$ satisfies (\ref{J18aa})--(\ref{J18ac}) and the constants $M$, $m$ and $L$ are defined by (\ref{A29a}), (\ref{Ma27a}) and (\ref{A29aa}). Let $p\in [-\varepsilon/2,0]$ and $\rho\leq\varepsilon/2$.
Then there exists constants $C_*$ and $\alpha_1\in (0,1)$  depending also on $C_1$, $L$, $M$, $m$, $\delta$ and $\rho$ such that
$$
||h||_{C^{3,\alpha_1}(\hat{Q}\cap B_{\rho/2}((q,p))}\leq C_*.
$$
\end{lemma}
The proof of this assertion is the same as the proof of Proposition 3 in \cite{KLN17}, where the results from \cite{GT}, \cite{LU} and \cite{ADN} are used. The important role plays by the inequality (\ref{Ma27aa}) in Lemma \ref{LA29a}  and its analog in $(q,p)$ variables:
\begin{equation}\label{Jun5a}
h_q^2+h_p^2\leq C_*\delta^{-2},\;\;h_p^2\geq C_*^{-1},
\end{equation}
where we have used also that $|\psi_y|\geq\delta$.  The main step here  is the local estimates  presented and proved in Sect. 2, Chapter 10, \cite{LU}, and due (\ref{Jun5a}) the conditions required for the local estimate are satisfied.

\begin{lemma}\label{LA29b} For each $\delta>0$ and $\varepsilon>0$ the set
\begin{equation}\label{MA1b}
{\mathcal K}_{\delta,\varepsilon}\cap {\mathcal O}
\end{equation}
is compact in ${\mathcal X}$.
\end{lemma}
\begin{proof} First we verify the following property.
Let $(\psi,\xi,\Lambda)$  satisfy (\ref{Ma13ba})--(\ref{Ma13bc}). Then
 there exists constants $C_*$ and $\alpha_1\in (0,1)$  depending  on $C_1$, $L$, $M$, $m$,  $\delta$ and $\varepsilon$ such that
\begin{equation}\label{Jun5aa}
||\psi||_{C^{3,\alpha_1}(\overline{D}_\xi)}\leq C_*,\;\;||\xi||_{C^{3,\alpha_1}(\Bbb R)}\leq C_*.
\end{equation}
%$$
%||\psi||_{C^{3,\alpha_1}(D_\xi\cap B_{\rho/2}((X,Y))}\leq C_*,\;\;||\xi||_{C^{3,\alpha_1}(S_\xi\cap B_{\rho/2}((X,Y))}\leq C_*.
%$$
The proof of this inequality (proved in Proposition 3.2 \cite{KL1}) follows from the lemma \ref{LMa31a} and the local estimates  from \cite{ADN}. Now the  required compactness property
in the space ${\mathcal X}$ follows from
(\ref{Jun5aa}) and compactness of inclusion of the H\"older spaces:
$$
{\mathcal K}_{\delta,\varepsilon}\cap {\mathcal O}\in C^{3,\alpha_1}(\overline{D_\xi})\times C^{3,\alpha_1}(\Bbb R),
$$
which leads to compactness of (\ref{MA1b}).
\end{proof}

To derive assertions (a) and (b) from Theorem \ref{TMa2a} and similar assertions from Theorem \ref{TMa2b} we proceed as follows.

(a) Assume that there exists $\varepsilon_*>0$  such all element of the branch belongs to $\widehat{\mathcal U}_{\varepsilon_*}$. Then we can apply Theorem \ref{ThJ3a}
for ${\mathcal K}_\delta=\overline{{\mathcal U}_{\delta,\varepsilon_*}}$ and ${\mathcal U}=\widehat{\mathcal U}_{\varepsilon_*}$, which leads to the proof of (a).

(b) Assume that (a) is not true. Then for every $\varepsilon$ there exists $\delta$ and
 an element $(\psi[t],\xi[t],\Lambda(t))\in {\mathcal U}_\delta$ on the branch such that
\begin{equation}\label{Jun15a}
 \varepsilon_*(\psi[t])<\varepsilon.
\end{equation}
One can choose $\varepsilon_j$ satisfying $\varepsilon_j\to 0$ as $j\to\infty$ and $(\psi[t_j],\xi[t_j],\Lambda(t_j))$  in such a way that $t_j\to\infty$ as $j\to\infty$ and the inequality (\ref{Jun15a}) still holds, when you include there the dependence on $j$. This proves the item (b) in Theorem \ref{TMa2a}.

%it does not belong to ${\mathcal U}_{\delta,\varepsilon}$. This means that %$\psi[t](x_*,y_*)<\varepsilon$
% and therefore there exist an element of the branch $(\psi[t],\eta[t],\lambda(t))$, %which do not belong to ${\mathcal U}_{\delta,\varepsilon}$. Now we choose  %sequences of $\varepsilon_j$ and $\delta_j$ tending to zero as $j\to\infty$ and %having the above property. We put
% \begin{equation*}
% \widehat{K}_j=\overline{{\mathcal U}_{\delta_j,\varepsilon_j}}.
% \end{equation*}
% Then
% $$
% \cup_j\widehat{K}_j={\mathcal U}
% $$
% and according Theorem \ref{ThJ3a} for every $j$ there exists $\hat{t}_j$ such that %the branch for $t>\hat{t}_j$ does not belong to $\widehat{K}_j$. Now we can choose %a sequence $\{t_j\}$ satisfying $t_j>\hat{t}_j$ and $t_j\to \infty$ as %$j\to\infty$. Now we take the function $(\psi_j[t_j],\xi[t_j],\Lambda(t_j))$ which %does not belong $\widehat{K}_j$. Then these functions satisfy the relation %(\ref{Jun7b}).
%This proves the item (b) in Theorem \ref{TMa2a}.

(iii) Let  $(\Psi_*(y),d_*)$ be the laminar solution  from Theorem \ref{TA17a} and let $\tau_*$ be the positive root of the dispersion equation (\ref{F21ac}). Then the solution
$\gamma(y;\tau_*)$ of the equation (\ref{F21ab}) is positive on the interval $(-d,0]$. By (\ref{A17acz}) the derivative $\psi_{\hat{x}}(\hat{x},y)$
is positive in $(0,\Lambda_*)\times [0,1)$. In the paper \cite{Koz1Loop} the density is a constant, but the basis of considerations there is Lemma 2.1 and the maximum principles, which can be found in Theorem 2.1, \cite{AN}, is still true. Therefore, all result there can be applied to this case also.

\section*{Acknowledgments}

This work was supported by the Ministry of Science and Higher Education of
the Russian Federation (agreement 075-15-2025-344 dated 29/04/2025 for
Saint Petersburg Leonard Euler International Mathematical Institute at PDMI RAS).

%\section{References}

{

\end{document}